\documentclass[11pt]{article}
\usepackage{authblk}
\usepackage[utf8]{inputenc}
\usepackage{graphicx}
\graphicspath{ {images/} }
\usepackage{caption}
\usepackage{subcaption}
\usepackage{float}
\usepackage{tikz}
\usepackage{geometry}
\usepackage{fancyhdr}
\usepackage{setspace}
\usepackage[style=numeric,sorting=nyt,url=false,maxnames=99]{biblatex}
\usepackage{graphicx}
\usepackage{hyperref}
\hypersetup{colorlinks,allcolors=blue}
\usepackage{amsthm,amsmath}
\usepackage{mathtools}
\usepackage{enumerate}
\usepackage{amsfonts}
\usepackage{bbold}
\usepackage{xcolor}
\usepackage{physics}
\newtheorem{theorem}{Theorem}
\newtheorem{lemma}[theorem]{Lemma}

\newtheorem{corollary}[theorem]{Corollary}

\theoremstyle{definition}
\newtheorem{definition}[theorem]{Definition}
\newtheorem{rmk}[theorem]{Remark}

\newtheorem{example}{Example}[section]

\numberwithin{equation}{section}  
\numberwithin{theorem}{section} 

\DeclareMathOperator{\co}{co}

\newcommand{\ep}{\varepsilon}

\title{Continuous-Time Quantum State Transfer with a Generalized Laplacian}
\author{Yujia Shi}
\affil{Department of Physics, Creighton University}

\date{\today}

\begin{document}

\maketitle

\begin{abstract}
Quantum walks generated by the adjacency matrix or the Laplacian are known to exhibit low transfer fidelity on general graphs. In this paper, we study continuous-time quantum walks governed by the generalized Laplacian operator \(L_k = A + k D\), where \(A\) is the adjacency matrix, \(D\) is the degree matrix, and \(k\) is a real-valued parameter. Recent work of Duda, McLaughlin, and Wong showed that in the single-excitation Heisenberg (XYZ) spin model, one can realize walks generated by this family of operators on signed weighted graphs. Motivated by earlier studies on vertex-weighted graphs, we demonstrate that for certain graphs, tuning the parameter \(k\) can significantly enhance the fidelity of state transfer between endpoints.
\end{abstract}

\paragraph{Keywords:} Quantum state transfer, Continuous-time quantum walks, Generalized Laplacian.
\paragraph{MSC: } 05C50

\section{Introduction}

Quantum state transfer on graphs, implemented through continuous-time quantum walks, has been widely studied as a model for quantum communication \cite{Bose_2003,Christandl_2005}. In this setting, the dynamics are governed by a Hamiltonian derived from the underlying graph, most commonly the adjacency matrix or the Laplacian. An algebraic perspective on this problem was developed in \cite{godsil2012state}. For the path graph \( P_n \), however, it is well known that these standard choices lead to poor transfer fidelities between the endpoints once the path becomes long. This limitation has motivated various modifications of the model, such as vertex-weighted graphs and alternative Hamiltonians.
In particular, this approach has been studied on several classes of graphs: for instance, by Kirkland and von Bommel on paths \cite{path}, 
by Lippner and the author on graphs with an involution \cite{p1}, 
and more generally in \cite{p2} for pairs of vertices with restricted forms of cospectrality.

In \cite{duda2025}, Duda, McLaughlin, and Wong showed that in the single-excitation Heisenberg (XYZ) spin model, one can realize continuous-time quantum walks on signed weighted graphs generated by the adjacency, Laplacian, and signless Laplacian operators, as well as the generalized Laplacian
\[
L_\alpha =A+(\alpha -1)D,
\]
where \( A \) is the adjacency matrix, \( D \) is the degree matrix, and \(\alpha\) is a real-valued parameter. This provides a natural framework for studying quantum walks driven by the one-parameter family of Hamiltonians \( H_\alpha = -\gamma L_\alpha \). Here, \(\gamma\) denotes the jumping rate of the walk, but since our focus is on transfer fidelity, we set \(\gamma = 1\) and consider walks generated by \( L_\alpha \) alone. 
For notational simplicity, we set $k := \alpha - 1$ and hence study the Hamiltonian \(H=-L_k\).
Our main result is that by tuning the parameter $k$, one can obtain quantum state transfer with high probability under the dynamics generated by $L_k$. 
We establish an explicit relation between $k$ and the transfer probability in a general setting.

The remainder of the paper is organized as follows. In Section~\ref{sec:prelim}, we review the necessary background, including definitions, spectral decomposition, and structural properties relevant to quantum state transfer. Section~\ref{sec:mainresult} discusses the existence of high-fidelity transfer using the generalized Laplacian and illustrates it with two examples.


\section{Preliminaries}
\label{sec:prelim}

\subsection{Quantum State Transfer on Graphs}

Let \( G = (V, E) \) be an undirected graph with \( n = |V| \) vertices. The Hilbert subspace of the system is \( \mathbb{C}^n \), with standard basis vectors \( \{e_u\}_{u \in V} \) corresponding to the vertices. The system evolves in continuous time according to the Schrödinger equation
\[
i \frac{d}{dt} \psi(t) = H \psi(t),
\]  
where \( H \in \mathbb{R}^{n \times n} \) is a Hermitian matrix, and here we adopt the standard convention $\hbar = 1$. The solution is
\[
\psi(t) = e^{-i H t} \psi(0).
\]

If the walk is initialized at vertex \( u \), i.e., \( \psi(0) = e_u \), then the probability of reaching vertex \( v \) at time \( t \) is
\[
P_{u \to v}(t) = \left| \langle e_v| e^{-i H t} |e_u \rangle \right|^2.
\]
For convenience, we denote the time-evolution operator by
\[
U(t) = e^{-i H t}.
\]
This is an $n \times n$ unitary matrix, and we write $U(t)_{u,v}$ for its $(u,v)$-entry. We say that \emph{perfect state transfer} (PST) from \( u \) to \( v \) occurs at time \( t > 0 \) if
\[
|U(t)_{u,v}| = 1.
\]
More generally, we are interested in \emph{high-fidelity transfer}, where the transition probability attains values arbitrarily close to 1 at certain times.

Let \( A \) be the adjacency matrix of \( G \), and \( D \) the degree matrix, i.e., \( D_{ii} = \deg(v_i) \). Two standard choices of Hamiltonian are:
\begin{itemize}
    \item \( H = -A \), the adjacency model, and
    \item \( H = -A + D \), the Laplacian model.
\end{itemize}

In this work, we consider the one-parameter family
\[
H = -(A + k D) , \quad k \in \mathbb{R}.
\]
As shown in \cite{duda2025}, the single-excitation Heisenberg (XYZ) spin model, with suitable coupling constants, can realize this generalized operator on signed weighted graphs.





\subsection{Conditions for High-Fidelity Transfer}\label{sec:sd}

Let the eigenvalues of \( L_k \) be \(\lambda_1,\lambda_2, \ldots,\lambda_n\)
with corresponding orthonormal eigenvectors \( \psi_1, \psi_2,\ldots, \psi_n \). Then
\[
L_k = \sum_{j=1}^n \lambda_j \psi_j {\psi_j}^\top.
\]
The time-evolution operator becomes
\[
U(t)_{u,v} = \sum_{j=1}^n e^{i \lambda_j t} \psi_j(u) \psi_j(v).
\]

Perfect state transfer is well known to be rare, requiring both vertices to be strongly cosepctral, together with eigenvalues satisfying the ratio condition, as first shown by Godsil \cite{godsil2012state}. Among the relaxations on perfect state transfer, a common scenario for achieving high transfer fidelity arises is when there exist two eigenvectors that are approximately
\(
(e_u \pm e_v)/\sqrt{2},
\)
so that \( U(t)_{u,v} \) depends primarily on these two terms rather than on all \( n \) eigenvalues. This approach was first discussed in the context of quantum tunneling~\cite{Tunneling_2012}, 
and was later adopted in the study of quantum state transfer~\cite{involution,path,p1,p2}.
More precisely, we seek a pair of eigenvectors \(\psi^{(1)}, \psi^{(2)}\) such that:
\begin{enumerate}
    \item \(\psi^{(1)}\) and \(\psi^{(2)}\) are strongly localized at \(u\) and \(v\),
    \item the ratios \(\psi^{(1)}(u)/\psi^{(1)}(v)\) and \(\psi^{(2)}(u)/\psi^{(2)}(v)\) are close to \(\pm 1\),
\end{enumerate}
\subsection{Graphs with Involutions}

Let \( \sigma: V \to V \) be an involution on a graph \( G \). The vertex set decomposes as
\[
V = N \cup S \cup \sigma(N),
\]
where \( S \) is the fixed-point set, and \( \sigma(N) \) is the image of \( N \) under \( \sigma \).

In this setting, the eigenvectors of \( H \) can be chosen to be either symmetric or antisymmetric with respect to \( \sigma \):
\[
[a \; b \; a] \quad \text{(symmetric)} \qquad \text{or} \qquad [c \; 0 \; -c] \quad \text{(antisymmetric)}.
\]
A detailed proof of this property can be found in \cite{involution}.
\begin{rmk}
The involution provides a quick method to identify pairs of cospectral vertices.  
In particular, when \(u = \sigma(v)\) lies outside the fixed-point set, the cospectrality is infinite, as in the case most relevant to this paper.
\end{rmk}
Another way to find a pair of strongly cospectral vertices is by counting the walks. In \cite{Tunneling_2012}, Lin, Yau, and Lippner define cospectrality and prove the following lemmas.
\begin{definition}
The \emph{cospectrality} between two vertices \(u,v \in V\), denoted \(\operatorname{co}(u,v)\), is the largest integer \(m\) such that for every \(k \leq m\), the number of closed walks of length \(k\) starting at \(u\) equals the number of closed walks of length \(k\) starting at \(v\).
\end{definition}

\begin{lemma}
If \(G\) admits an involution \(\sigma\) and \(u = \sigma(v)\) for some pair of vertices not in the fixed-point set, then \(\operatorname{co}(u,v) = \infty\).
\end{lemma}
\begin{lemma}
    If \(\operatorname{co}(u,v) = \infty\), then for every eigenvector, \(\psi_j(u)=\pm\psi_j(v)\)
\end{lemma}
\section{High-Fidelity Transfer with $L_k$}\label{sec:mainresult}

In \cite{p2}, we provided a method for assigning positive self-loop weights  \(Q\) to the starting vertex \(u\) and the target vertex \(v\) that guarantees high transfer fidelity. In particular, the peak fidelity
\[
F(Q) := \sup_{t > 0} |U(t)_{u,v}|
\]
can be made arbitrarily close to 1.

\begin{theorem}\label{thm:p2}\cite{p2} Let $G(V,E)$ be a finite graph with maximum degree $m$, and $u,v \in V$ fixed. Assume $c \geq d$ where $c = \co(u,v)$ denotes their cospectrality and $d = d(u,v)$ denotes the distance between $u$ and $v$. For any $\ep >0$ we have 
\[ Q > 16 \frac{1}{\ep^{1/\min(2,c-d+1)}} m^{1+\max\left(\frac{1}{2}, \frac{d}{c-d+1}\right)}  \Longrightarrow F(Q) > 1- \ep. \]
And its readout time has an upper bound $t_0 < 2\pi (Q+m)^{d-1}$
\end{theorem}
In essence, Theorem~\ref{thm:p2} shows that when the self-loop weight $Q$ is sufficiently large, 
both conditions from Section~\ref{sec:sd} are satisfied, ensuring that the peak fidelity $F(Q)$ 
can be made arbitrarily close to 1.

\begin{corollary}\label{cor:revised}
Theorem~\ref{thm:p2} also holds for negative values of \(Q\), with the same bound applied to \(|Q|\).
\end{corollary}

\begin{proof}
In the proof of Theorem~\ref{thm:p2}, we use the observation that when $Q>0$, the two largest eigenvalues lie in the interval $[Q-m,\,Q+m]$, separated from the remaining eigenvalues contained in $[-m,m]$. When $\lambda$ is sufficiently large compared with $m$, the two 
conditions from Section~\ref{sec:sd} are met. 

For negative $Q$, the same reasoning applies. 
Although the relevant eigenvalues now lie below $-m$, the convergence argument is unchanged once we replace $\lambda$ by $|\lambda|$. 
To be more specific, for vertices $x,y$, let $\mathcal{P}_{xy}$ denote the set of walks from $x$ to $y$ avoiding $u,v$ except possibly at the endpoints, and define
\[
Z_{xy}(\lambda) = \sum_{P \in \mathcal{P}_{xy}} \lambda^{-|P|}
   = \sum_{k=d(x,y)}^\infty \frac{n_k(xy)}{\lambda^k},
\]
where $n_k(xy)$ is the number of such walks of length $k$.  
Since $n_k(xy)\le m^k$ in a graph of maximum degree $m$, we obtain the bound
\[
Z_{xy}(\lambda) \;\le\; 
\Bigl(\tfrac{m}{|\lambda|}\Bigr)^{d(x,y)} 
\frac{1}{1-\tfrac{m}{|\lambda|}},
\]
which shows that $Z_{xy}(\lambda)$ converges absolutely whenever $|\lambda|>m$. 
Thus the convergence estimates used in the proof remain valid for negative $Q$, 
and the argument proceeds exactly as before, yielding the same lower bound expressed in terms of $|Q|$.

\end{proof}

\begin{proof}
In the proof of Theorem~\ref{thm:p2} in ~\cite{p2}, we used the observation that when $Q>0$, 
the two largest eigenvalues lie in the interval $[Q-m,\,Q+m]$, separated from the 
remaining eigenvalues contained in $[-m,m]$. 
When $\lambda$ is sufficiently large relative to $m$, the eigenvectors exhibit 
the two behaviors described in Section~\ref{sec:sd}.

For negative $Q$, the same reasoning applies. 
Although the relevant eigenvalues then lie below $-m$, the convergence argument 
is unchanged once we replace $\lambda$ by $|\lambda|$. 
Indeed, for vertices $x,y$, let $\mathcal{P}_{xy}$ denote the set of walks from 
$x$ to $y$ avoiding $u,v$ except possibly at the endpoints, and define
\[
Z_{xy}(\lambda) = \sum_{P \in \mathcal{P}_{xy}} \lambda^{-|P|}
   = \sum_{k=d(x,y)}^\infty \frac{n_k(xy)}{\lambda^k},
\]
where $n_k(xy)$ is the number of such walks of length $k$.  
Since $n_k(xy)\le m^k$ in a graph of maximum degree $m$, we obtain the bound
\[
Z_{xy}(\lambda) \;\le\; 
\Bigl(\tfrac{m}{|\lambda|}\Bigr)^{d(x,y)} 
\frac{1}{1-\tfrac{m}{|\lambda|}},
\]
which shows that $Z_{xy}(\lambda)$ converges absolutely whenever $|\lambda|>m$. 
Thus the convergence estimates from the positive case remain valid, and the argument 
proceeds exactly as before, yielding the same lower bound, now expressed in terms of $|Q|$.
\end{proof}

\begin{lemma}\label{lemma:-2I}
Suppose a graph \(G\) has \(n\) vertices, with two vertices \(u,v\) of degree \(d_1\), and all other vertices of degree \(d_2\neq d_1\). For a vertex \(x\), let \(e_x\) denote the standard basis vector, and denote 
\(E_x := e_x e_x^\top\), the projection matrix. 
Then the quantum walk generated by the generalized Laplacian \(L_k\) is equivalent (up to a global phase) to a walk generated by
\[
A + Q (E_u + E_v),
\]
where \(Q = k(d_1 - d_2)\).
\end{lemma}

\begin{proof}
Since subtracting \(k d_2 I\) changes the evolution only by a global phase, it does not affect transfer probabilities. The remaining term is
\[
A + k(D - d_2 I) = A + k(d_1 - d_2)(E_u + E_v),
\]
which is the desired form with \(Q = k(d_1 - d_2)\).
\end{proof}

\begin{example}
Consider the complete bipartite graph \(K_{2,n-2}\) with \(n \geq 5\).  
(The special case \(n=3\) is simply the path \(P_3\), which will be treated separately.)  
In \(K_{2,n-2}\), the two vertices \(u,v\) in the part of size 2 each have degree \(n-2\),  
while every vertex in the other part has degree 2.  

It is easy to verify that \(K_{2,n-2}\) admits an involution, and moreover \(\co(u,v)=\infty\).  
Thus, by Lemma~\ref{lemma:-2I}, the Hamiltonian for the quantum walk generated by the generalized Laplacian $L_k$ is equivalent to
\[
A + Q(E_u+E_v), 
\qquad 
Q = k (n-4).
\]

By Theorem~\ref{thm:p2} together with Corollary~\ref{cor:revised}, for \(n \geq 5\) it suffices to require

\[
|k| \;>\; 16\,\frac{(n-2)^{3/2}}{\sqrt{\epsilon}\,(n-4)}.
\]
When this bound holds, the endpoint transfer fidelity is at least $1-\epsilon$.

\end{example}

\begin{example}
If \(G\) is the path \(P_n\) and \(u,v\) are its endpoints, then \(u\) and \(v\) are cospectral.  
In this case we have \(Q = -k\), and the condition simplifies to
\[
|k| \;>\; \frac{32\sqrt{2}}{\sqrt{\epsilon}}.
\]

To illustrate, we compare the endpoint transfer fidelity on the path $P_6$ under four Hamiltonians: the three standard choices and the generalized Laplacian with $k=143$, 
chosen according to the lower bound on $|k|$ when $\epsilon = 0.1$. Note that the corresponding readout time is large, consistent with the theorem~\ref{thm:p2}, 
which predicts scaling on the order of $\mathcal{O}(Q^d)$.

\begin{figure}[h]
\centering
\begin{subfigure}{0.43\textwidth}
    \centering
    \includegraphics[width=\linewidth]{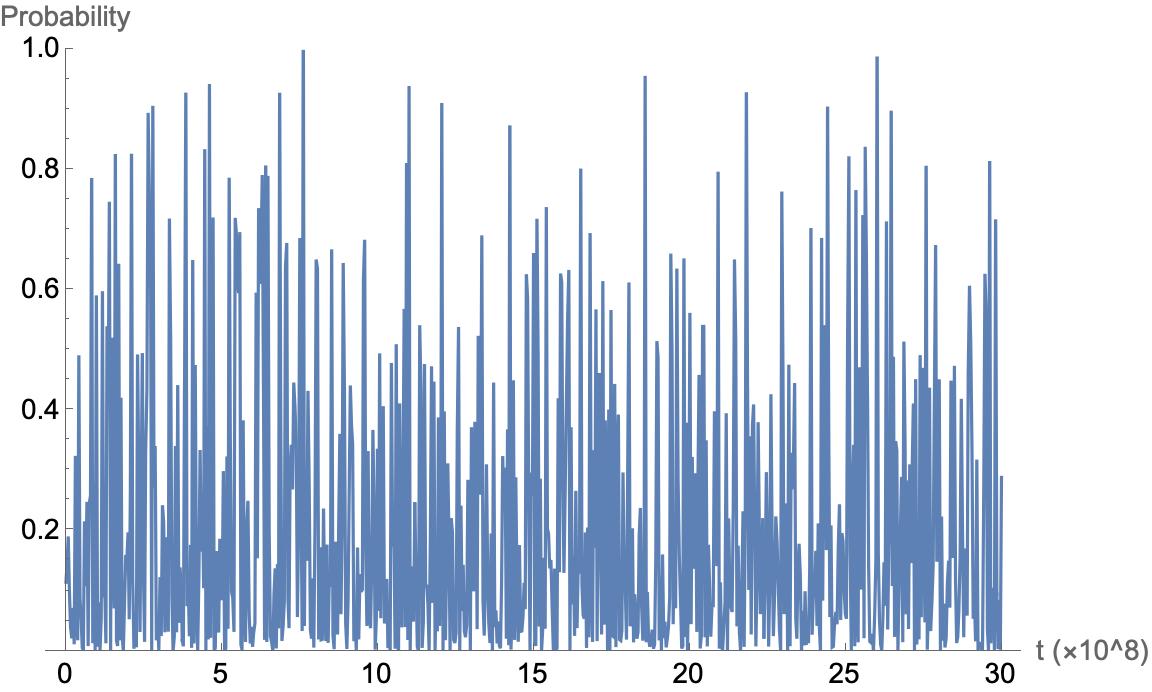}
    \caption{Adjacency matrix}
\end{subfigure}
\hfill
\begin{subfigure}{0.43\textwidth}
    \centering
    \includegraphics[width=\linewidth]{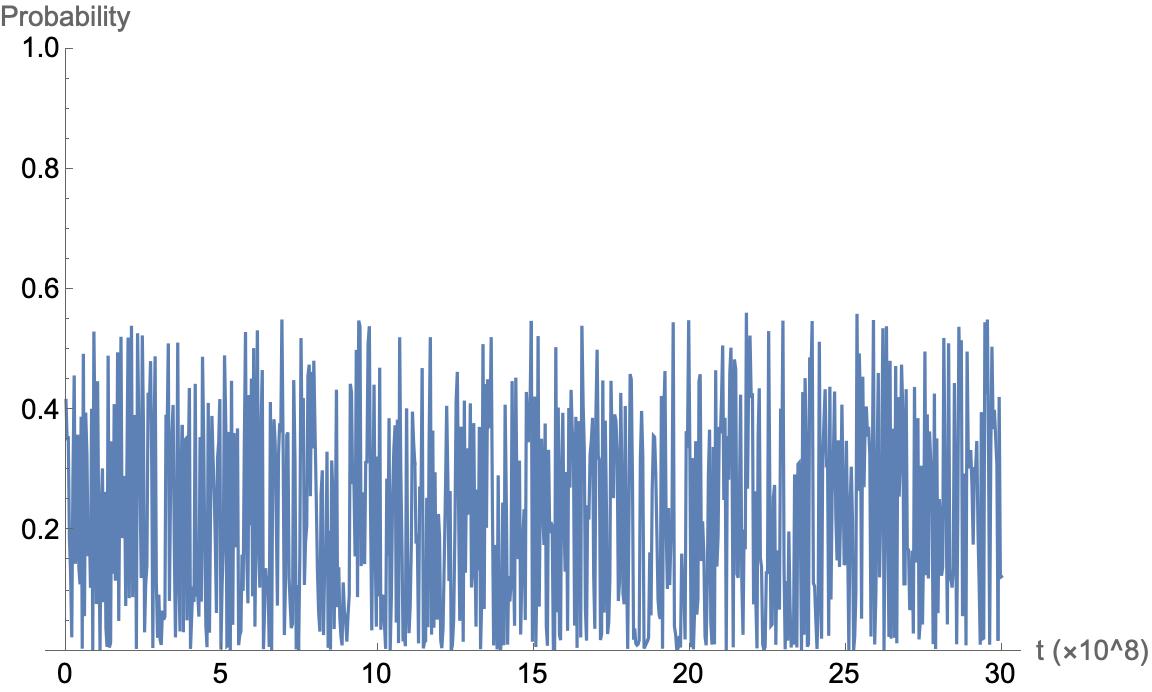}
    \caption{Laplacian}
\end{subfigure}
\hfill
\begin{subfigure}{0.43\textwidth}
    \centering
    \includegraphics[width=\linewidth]{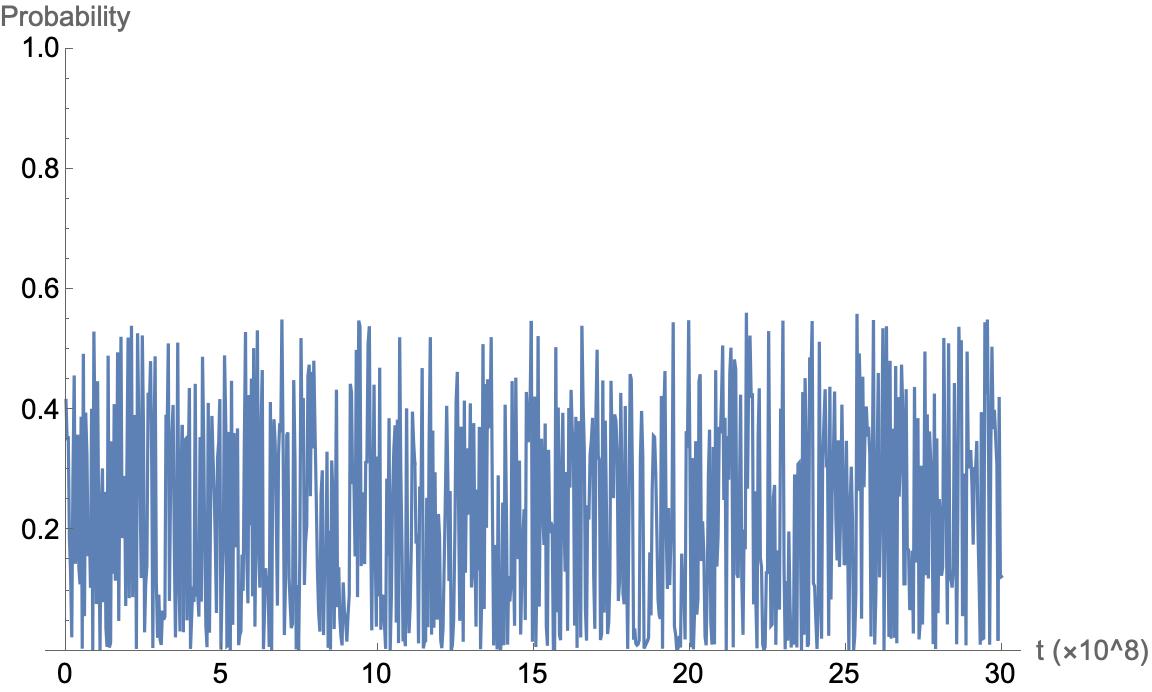}
    \caption{Signless Laplacian}
\end{subfigure}
\hfill
\begin{subfigure}{0.43\textwidth}
    \centering
    \includegraphics[width=\linewidth]{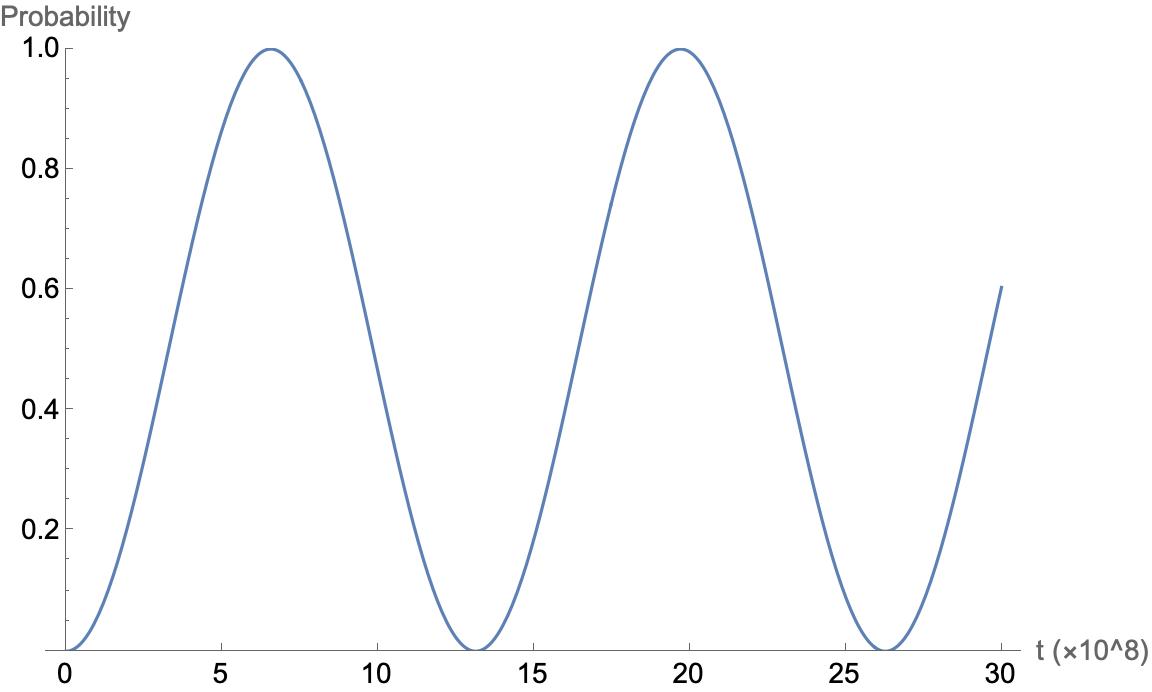}
    \caption{Generalized Laplacian, with $k=143$}
\end{subfigure}
\caption{Endpoint transfer fidelity on \(P_6\) under different Hamiltonians.}
\label{fig:p5}
\end{figure}
\end{example}

\section*{Conclusion}

This work provides an initial investigation into quantum state transfer driven by the one-parameter family of generalized Laplacians $L_k = A + kD$. 
We showed that by tuning the parameter $k$, one can obtain high-probability transfer between designated vertices, 
even in graph families where the standard adjacency and Laplacian Hamiltonians fail to do so. 
The mechanism relies on the presence of a pair of highly cospectral vertices, which allows two eigenvectors to become nearly symmetric when $k$ is large. 
This effectively reduces the dynamics to a two-dimensional subspace and enables reliable transfer with high fidelity. 
In addition, we extended earlier results on vertex-weighted graphs to obtain an explicit relation between fidelity and $k$ for all real values of $k$.

\section*{Acknowledgments}

Thanks to Thomas G.~Wong for suggesting the problem of state transfer with the generalized Laplacian. This material is based upon work supported in part by the National Science Foundation EPSCoR Cooperative Agreement OIA-2044049, Nebraska’s EQUATE collaboration. Any opinions, findings, and conclusions or recommendations expressed in this material are those of the author(s) and do not necessarily reflect the views of the National Science Foundation.

\printbibliography
\end{document}